\documentclass[preprint,journal]{vgtc}              %

\preprinttext{StorySets: Ordering Curves and Dimensions for Visualizing Uncertain Sets and Multi-Dimensional Discrete Data}
\manuscriptnote{}

\vgtccategory{Research}

\vgtcpapertype{algorithm/technique}

\title{StorySets: Ordering Curves and Dimensions \mbox{for Visualizing Uncertain Sets and Multi-Dimensional Discrete Data}}

\author{%
  \authororcid{Markus Wallinger}{0000-0002-2191-4413}\and
  \authororcid{Annika Bonerath}{0000-0002-8427-3246}\and
  \authororcid{Wouter Meulemans}{0000-0002-4978-3400}\and
  \authororcid{Martin Nöllenburg}{0000-0003-0454-3937}\and
  \authororcid{Stephen Kobourov}{0000-0002-0477-2724}\and
  \authororcid{Alexander Wolff}{0000-0001-5872-718X}
}

\authorfooter{

}

\abstract{%
We propose a method for visualizing uncertain set systems, which differs from previous set visualization approaches that are based on certainty (an element either belongs to a set or not). Our method is inspired by storyline visualizations and parallel coordinate plots: (a) each element is represented by a vertical glyph, subdivided into bins that represent different levels of uncertainty; (b) each set is represented by an x-monotone curve that traverses element glyphs through the bins representing the level of uncertainty of their membership. Our implementation also includes optimizations to reduce visual complexity captured by the number of turns for the set curves and the number of crossings. Although several of the natural underlying optimization problems are NP-hard in theory (e.g., optimal element order, optimal set order), in practice, we can compute near-optimal solutions with respect to curve crossings with the help of a new exact algorithm for optimally ordering set curves within each element's bins. With these optimizations, the proposed method makes it easy to see set containment (the smaller set's curve is strictly below the larger set's curve). A brief design-space exploration using uncertain set-membership data, as well as multi-dimensional discrete data, shows the flexibility of the proposed approach. 
}
  
\graphicspath{{figs/}{figures/}{pictures/}{images/}{./}} %

\teaser{
     \centering
     \subfloat{
         \includegraphics[width=\textwidth]{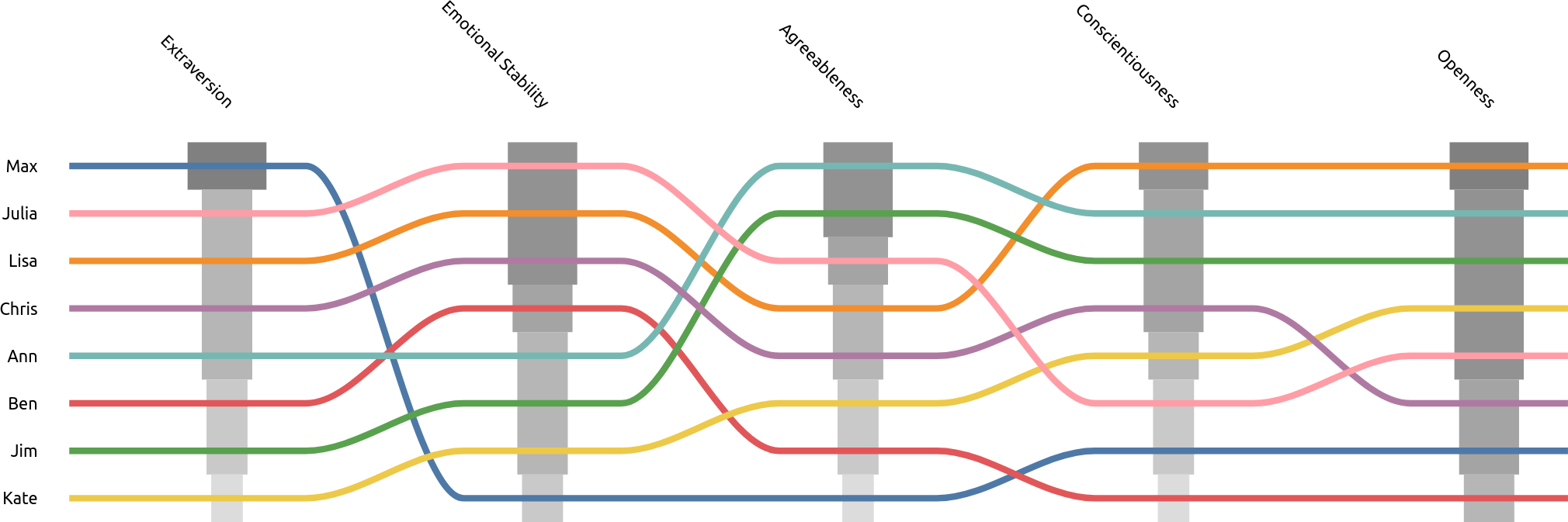}
     }
     \caption{People and their (uncertain) association with character traits visualized with the StorySets approach. Each x-monotone curve represents a person (set) and each vertical glyph a character trait (element). Color and thickness of the glyphs encode (binned) certainty ranging from 1 (certainly a set member/character trait) to 0 (certainly not a set member/character trait). The element-set membership uncertainty is encoded by the position the set-curve passes through the element-glyph's respective bin. The StorySets approach minimizes crossings and optimizes the spacing of the curves, thus avoiding overplotting and making curves traceable.}
    \label{fig:teaser}
}

\usepackage{tabu}                      %
\usepackage{booktabs}                  %
\usepackage{mwe}    
\usepackage{mathptmx}                  %
\usepackage{amsmath,amsthm}
\usepackage{xspace}
\usepackage[basic]{complexity}

\usepackage[vlined,algoruled]{algorithm2e}

\newtheorem{theorem}{Theorem}

\begin{document}
\firstsection{Introduction}
\maketitle

Representing discrete, multi-dimensional relational data as set systems is a very versatile approach from both a modeling and a visualization perspective. 
A set system consists of a ground set of elements and a family of
different subsets of these elements.
A set system can also be seen as a hypergraph whose vertices
are the elements of the ground set and whose (hyper)edges are the subsets.
Hence, a set system where all subsets have cardinality~2 represents a
standard, undirected graph.
More generally, a set system may represent multi-dimensional data with
attributes.  For example, patients can be the elements, and clinical symptoms can be the sets, which are formed by groups of patients who show a specific
symptom. 

When analyzing such set-typed data, questions and tasks can be \emph{element-based},
e.g., ``Which symptoms does patient A show?'', or \emph{set-based},
e.g., ``What other symptoms do patients with symptom $S$ have?''
There are several different approaches for visualizing set systems; see Alsallakh et al.~\cite{alsallakh2016} or Hu et al.~\cite{hmnr-svds1-18}.
However, in many real-world examples, the membership of elements in
sets is not just a binary yes/no, but can be subject to uncertainties
of different types: from noisy data acquisition to data
transformations, and vagueness in human interpretation.
In the example above, the presence or absence of a symptom in a
patient 
is likely not always 0\% or 100\%.  
This turns the set system into an \emph{uncertain} (or \emph{fuzzy}) set system, where the membership of each element in a set can be quantified as a value $p$ between~0 (certainly not a set member) and~1 (certainly a set member).
Despite the abundance of such uncertain set-typed data in the real world, there is a surprising lack of research on visualizing uncertain set systems with a focus on representing the degree of uncertainty in the set membership information~\cite{TominskiBBFMMP23}. 
Possible reasons can be found both in the difficulty of quantifying and controlling such uncertainties and in the difficulty of visually representing uncertainty in a way that can be correctly perceived and interpreted by humans~\cite{jena2020uncertainty,bonneau2014overview}.

In this paper, we present a design space exploration for the representation of set data with uncertain set memberships, 
based on %
existing approaches for representing sets, graphs, and multi-dimensional data without uncertainties. 
From this exploration, we develop a novel %
visualization style, called StorySets, which combines aspects from storyline visualizations~\cite{tanahashi2012design} and parallel coordinate plots~\cite{HeinrichW13}.
In StorySets (see \cref{fig:teaser} for an example), sets are represented as x-monotone curves similar to the character curves in a storyline visualization or the subway lines along a track served by multiple lines. 
We represent the elements by a linear sequence of vertical glyphs
(for which we propose several variants), crossed by the curves
of the sets such that the vertical position at which a curve crosses
an element indicates the certainty of that element's membership in the
respective set.
Since the visual quality of a StorySets representation depends %
on the complexity of the curve trajectories for the different sets (e.g.,  how entangled they are, how many crossings they induce), we use the available degrees of freedom to optimize their appearance. 
To this end, (1)~we may rearrange the linear order in which the
different element glyphs are positioned and (2)~among the different
sets that contain a particular element with the same certainty, we may
locally rearrange the vertical order in which the set curves cross that
element's glyph.
We propose a 2-step algorithm that first finds a good element ordering using a traveling salesperson model and then finds an optimal set ordering that obtains the minimum number of crossings for the given element ordering. 
By iterating the two steps for a few iterations, we can compute
high-quality StorySets visualizations with few crossings and a few curve
turns in a relatively short time.  To support this claim, we perform a comprehensive computational
experiment where we compare the performance of different optimization variants on quality metrics and runtime.
Lastly, we present a case study with two real-world examples where we discuss how a StorySets representation can potentially help visualize and communicate uncertain set data and multi-dimensional discrete data.

\section{Related Work}
Even though set data naturally models many real-world problems and uncertainty is a fact of life, there is very little work on uncertainty in set visualization, as summarized in a recent paper proposing a conceptual framework on this topic~\cite{TominskiBBFMMP23}.
Publications related to the methods presented in this paper cover several different areas, but here we focus on three of particular interest: visual representation of sets, parallel coordinates, star plots and storylines, and visualizing uncertainty. 

\subsection{Set Visualization} Set visualization approaches such as Venn diagrams and Euler diagrams represent each set with a closed curve and elements with points inside or outside these curves; they date back hundreds of years. 
Methods for automatically generating Euler diagrams represent the sets with circles and ellipses~\cite{sfrh-adedwc-12,w-eaacved-12,larsson2020eulerr,mr-efled-14}, or as well with less regular shapes~\cite{rd-ued-10,saa-favos-09,sas-sabied-16,srhz-iged-11}. 
While Euler and Venn diagrams are intuitive, they do not scale well for large numbers of sets or elements~\cite{alsallakh2016}. 

In overlay-based set visualizations, elements are placed at pre-defined positions in the plane, and sets are represented by regions.
Bubble Sets~\cite{cpc-bsrrwioev-09} create isocontours for each set, which enclose their corresponding points. LineSets~\cite{ahrc-dslnvt-11} connect the elements of each set by a smooth and short curve reminiscent of lines in metro maps.
MetroSets~\cite{jwkn-mvsmm-20}, computes a path support graph from an abstract set system and visualizes it in the style of a metro map, visually similar to LineSets.

Matrix-based visualizations map sets to rows and elements to columns and indicates containment with symbols or glyphs. Several interactive systems are based on this idea: ConSet~\cite{kls-vcwpmd-07}, OnSet~\cite{smds-ovtlbd-14}, UpSet~\cite{lex2014upset}.
Similarly, in linear
diagrams~\cite{rsc-vswld-TCHI15,lm-clmdv-19,stapleton2019efficacy,dn-cold-22, wallinger23,RodgersCBNWD24}
each set is a row in the matrix indicated by one or more horizontal
line segments.
The utility of matrix-based representations in general and linear diagrams, in particular, depends on carefully ordering rows and columns of the matrix.

Potentially, some of the above-mentioned approaches could be adapted to uncertain set data. For example, matrix-based approaches could indicate the respective uncertainty by color or size of glyphs~\cite{TominskiBBFMMP23}. However, adapting the visual representation is not always straightforward. While StorySets is able to handle certain set data as well, the focus is clearly on explicitly representing uncertain membership, where bins represent uncertainty classes that attempt to not over-complicate the representation.

\subsection{Parallel Coordinates, Star Plots and Storylines}

Parallel coordinate plots visualize high dimensional data by generalizing 2D and 3D Cartesian plots; see the state-of-the-art report~\cite{HeinrichW13}. Multiple axes (one for each dimension) are represented by parallel vertical lines, and each data point as an $x$-monotone curve intersecting the axes at the corresponding $y$ values. 

Parallel coordinates were initially designed to represent continuous data~\cite{inselberg1990parallel}, but were soon applied to discrete and categorical data~\cite{rosario2004mapping}.
ParallelSets~\cite{BendixKH05} adopts the layout
of parallel coordinates but substitutes the individual data points by
a frequency-based representation.
As with matrix-based set visualization, the order of the axes in a parallel-coordinates plot has an 
impact on the resulting visualization.  Different metrics to evaluate the order of axes have been proposed, based on correlations~\cite{ferdosi2011visualizing}, dimension clustering~\cite{ankerst1998similarity}, data clustering~\cite{tatu2010automated}, outlier detection~\cite{wilkinson2006high}, etc. As many of the corresponding underlying optimization problems are NP-hard, new heuristics and approaches, including machine learning, continue to be proposed~\cite{LHH12,LHZ16}.

In star plots, the axes of a parallel coordinates plot become the spokes of a wheel, and each data point becomes a closed curve, or a ``star"~\cite{chambers1983graphical}. To deal with more data points and more dimensions, variants such as DataRoses have been proposed that combine interaction and summarization~\cite{elmqvist2008datameadow}.

Dating back to 2009's XKCD comic 657 ``Movie Narrative Charts,''
storyline visualizations have been used to show a timeline 
of interactions between a set of characters, represented by a set of curves. 
The design space of the storylines visualization metaphor has been
explored~\cite{tanahashi2012design}, and ways to improve readability
have been proposed~\cite{liu2013storyflow}.  Storylines have also
motivated algorithmic work about line crossings and curve
complexity~\cite{pupyrev11,fpw-omlbbc-JGAA15,knpss-mcsv-GD15,gjlm-cmsv-GD16,dfflmrsw-bcsv-JGAA17,dlmw-csfbc-GD17}.

One issue of parallel coordinate plots is overplotting. This is especially obvious when discrete data is mapped to axes. This has been investigated in the related work, for example, by filtering or aggregating data, spatial distortion, or reordering of axes. StorySets is similar to parallel coordinate plots in the sense that sets are represented as x-monotone curves that pass through bins in glyphs representing set elements and their respective certainty. StorySets also optimizes element order, which is similar to reordering axes. Additionally, by binning, we introduce a degree of freedom that allows us to reorder curves. %
This makes it possible to avoid some curve crossings %
and to introduce spaces between curves. To our knowledge, such an approach has not been previously investigated.    

\subsection{Uncertainty Visualization}\label{sec:uncertainty_related}
There are several surveys on uncertainty visualization~\cite{bonneau2014overview,hullman2018pursuit} and the authors of the most recent one~\cite{jena2020uncertainty} point out a ``gulf between the rhetoric in visualization research about the significance of uncertainty and the inclusion of representations of uncertainty in visualizations used in practice.'' 
While there are already hundreds of papers about uncertainty in general, there are very few that even consider uncertainty in set visualization in particular~\cite{TominskiBBFMMP23}.
Fuzzy set visualization, using opacity-varying freeform Euler-like diagrams~\cite{ZhuXLS18}, represents membership uncertainty with distance from the region's center and does not generalize to overlapping sets. A graph-based model for visualizing fuzzy overlapping communities in networks~\cite{VehlowRW13} assumes that each element belongs to multiple communities, and the sum of memberships must add to 1. 
 In the context of the set system, uncertainty can affect the elements as well as the sets. It can arise for several of the usual reasons: ambiguity, limited or missing data, noisy signals, or human impreciseness.
In this paper, we focus on uncertainty in element-set membership. To the best of our knowledge, our proposed method, StorySets, is the first attempt to visualize uncertain set systems.

\section{Design Space Exploration}

\begin{figure}[t]
  \begin{subfigure}[b]{0.5\linewidth}
    \centering
    \includegraphics[width=\textwidth, page=1]{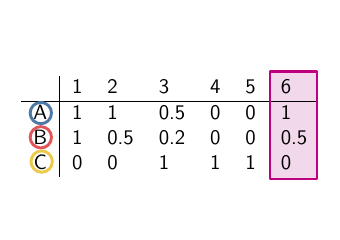}      
    \caption{matrix based}
    \label{fig:designspace:matrix}
  \end{subfigure}
  \hfill
  \begin{subfigure}[b]{0.5\linewidth}
    \centering
    \includegraphics[width=\textwidth, page=2]{figures/designexploration.pdf}
    \caption{Euler diagram}
    \label{fig:designspace:euler}
  \end{subfigure}

  \begin{subfigure}[b]{0.5\linewidth}
    \centering
    \includegraphics[width=\textwidth, page=3]{figures/designexploration.pdf}      
    \caption{LineSet}
    \label{fig:designspace:lineset}
  \end{subfigure}
  \hfill
  \begin{subfigure}[b]{0.5\linewidth}
    \centering
    \includegraphics[width=\textwidth, page=4]{figures/designexploration.pdf}      
    \caption{Linear Diagram}
    \label{fig:designspace:lineardiag}
  \end{subfigure}

  \begin{subfigure}[b]{0.5\linewidth}
    \centering
    \includegraphics[width=\textwidth, page=5]{figures/designexploration.pdf}      
    \caption{Star plot}
    \label{fig:designspace:starplot}
  \end{subfigure}
  \hfill
  \begin{subfigure}[b]{0.5\linewidth}
    \centering
    \includegraphics[width=\textwidth, page=6]{figures/designexploration.pdf}      
    \caption{StorySet}
    \label{fig:designspace:storyset}
  \end{subfigure}
    
  \caption{(a) The input: the (un)certainty of an element being a member of a set; (b--f) different visual representations of an uncertain set system
    with three sets $\{A, B, C\}$ and six elements $\{1, 2, 3, 4, 5, 6\}$.} 
  \label{fig:designspace}
\end{figure}

We begin with a high-level overview of set visualization techniques and a discussion about how they could be extended to express uncertainty. Then, we discuss the StorySets design space in more detail.

\begin{itemize}
    \item \textbf{Matrix-based} approaches are well-established techniques for visualizing set systems.  Each row corresponds to a set and each column to an element; see \cref{fig:designspace:matrix}. We suggest that the matrix is filled with the values of the set membership certainties. A similar variant was proposed by Tominski et al.~\cite{TominskiBBFMMP23} where they place glyphs in the matrix that encodes the uncertainty with size and/or color.         
    \item \textbf{Euler Diagrams} are another popular visualization technique for set systems. We suggest representing each element as a unit-sized disk (instead of a point) and indicating how certain the membership of an element in a set is by the element's distance to the set boundary or with partial containment of the element in the set's region; see \cref{fig:designspace:euler}. Zhu et al.~\cite{ZhuXLS18} use a similar idea for a single set. 
    A similar strategy could show uncertainty in LineSets; see \cref{fig:designspace:lineset}. 
    \item \textbf{Linear diagrams} are similar to matrix-based approaches where each set represents a row in a matrix and each element represents a column in a matrix. Membership of elements in sets is indicated by one or more horizontal line segments that pass through all contained elements~\cite{rsc-vswld-TCHI15}. We suggest to encode the set membership uncertainty by the line width; see   \cref{fig:designspace:lineardiag}.
    \item \textbf{Star plots}~\cite{chambers1983graphical} represent
    elements are partial spokes of a wheel and sets are closed curves that either go through or avoid the spokes. 
We suggest encoding uncertainty by the distance of a curve from the center of the diagram; see~\cref{fig:designspace:starplot}.
    \item \textbf{Storylines} are the last concept we want to highlight for the visualization of uncertain sets; see \cref{fig:designspace:storyset}. In a Storyline~\cite{knpss-mcsv-GD15,dfflmrsw-bcsv-JGAA17} visualization, each set corresponds to an x-monotone curve, and glyphs represent elements. In our adaptation, we would place said glyphs next to each other and anchor them at a top baseline. Each glyph encodes the certainty levels, e.g., by a set of stacked boxes where the box width corresponds to the certainty level. The curve of a set passes horizontally through the glyph at the certainty level that corresponds to the set membership certainty.  
\end{itemize}
While all of the mentioned visualizations of uncertain sets have their advantages and disadvantages, we aim for a visualization that provides interpretable set membership uncertainty and shows set containment. It is known that Euler diagrams do not scale well with many sets or elements,  even without uncertainties~\cite{alsallakh2016}. Matrix-based approaches and linear diagrams make it difficult to perceive set containment~\cite{alsallakh2016}. Hence, we deem that the adaptions of star plots and storylines best suit these tasks. Since they are computationally and visually similar, we summarize both under the term StorySets: on the one hand, with a star layout and on the other hand, with a storyline layout. In the following, we discuss the design space of StorySets.

\begin{figure}
    \centering    
    \begin{subfigure}{0.45 \linewidth}        
        \includegraphics[page=4, width=\textwidth]{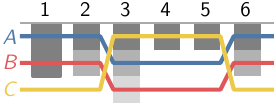}
        \caption{glyph rendering: color gradient}
        \label{fig:variants:glyphs3}
    \end{subfigure} 
    
    \begin{subfigure}{0.45 \linewidth}        
        \includegraphics[page=5, width=\textwidth]{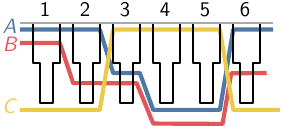}  
        \caption{bin heights: uniform}
        \label{fig:variants:uniformbinheights}
    \end{subfigure} 
    \begin{subfigure}{0.45 \linewidth}        
        \includegraphics[page=6, width=\textwidth]{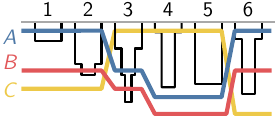}  
        \caption{bin heights: given distribution}
        \label{fig:variants:givenbinheights}
    \end{subfigure} 
    \begin{subfigure}{0.45 \linewidth}        
        \includegraphics[page=7, width=\textwidth]{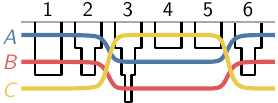}  
        \caption{rounded curves}
        \label{fig:variants:roundedlines}
    \end{subfigure} 
    \begin{subfigure}{0.45 \linewidth}        
        \includegraphics[page=8, width=\textwidth]{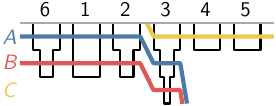}  
        \caption{horizontal compactness}
        \label{fig:variants:verticalcompactness}
    \end{subfigure} 
    \caption{Different design variants of StorySets.}
    \label{fig:variants}
\end{figure}

\paragraph{Storyline and Star Layout}
As a first design dimension, we want to compare the storyline and star layouts; see \cref{fig:designspace:starplot} and \cref{fig:designspace:storyset}. With a storyline layout, we obtain a visualization that is more similar to Storylines and linear diagrams while the star layout yields a star plot visualization. From a computational point of view, both variants are similar, with the exception that for the star layout, one needs to take care of the additional cyclic constraints between the last and the first element. A storyline layout can handle many elements and sets, but the star layout has a better aspect ratio and is more compact.
By varying the size of the inner `free' circle for the star layout, we can fit more elements.  

\paragraph{Glyph Rendering}
The rendering of the element symbols is crucial for perceiving set membership certainties. We focus on two visual encodings: (i) the symbols consist of a set of stacked boxes where the widths of the boxes encode the certainty level; see \cref{fig:designspace:storyset}, and (ii) the symbol is a set of boxes of the same width with different colors that encode the certainty level; see \cref{fig:variants:glyphs3}. We are not aware of relevant perception studies that might guide these (somewhat specific) choices.

\paragraph{Box Heights}
We consider three options for choosing the heights of the uncertainty boxes. A straightforward approach is to have boxes of uniform height; see \cref{fig:variants:uniformbinheights}. This would make it easy to compare the certainty levels of different elements since a change in the y-value of a curve directly corresponds to a change in the certainty level. A downside is that this can lead to large symbols. 
A second option is to scale each box of each element according to the number of curves going through it; see \cref{fig:designspace:storyset}. This makes the symbol sizes smaller but loses the direct correspondence of uncertainty and y-value. This approach highlights the number of uncertain set containments of an element (its symbols' height).
The third option is to scale the boxes of an element according to a given distribution; see \cref{fig:variants:givenbinheights}. For example, if the displayed curves are just a selection of a larger number of sets, one can scale the boxes according to the overall number of sets that contain the element. This variant is closely related to stacked bar charts. It has the advantage that the visualization encodes even more information than just the uncertain set containments of the displayed curves. 

\paragraph{Curve Smoothness}
Another visual parameter of our StorySets is the smoothness of the curves. In a straightforward implementation, one would consider each curve as a series of straight-line segments that have their bends in the symbols; see \cref{fig:designspace:storyset}. Curves that have smooth bends (e.g., cubic splines) are more appealing but might introduce unnecessary crossings. We opt for a smooth poly-line variant (rounded curves) that avoids this problem; see \cref{fig:variants:roundedlines}.

\paragraph{Compactness}
Since for large set systems, the visualization can become unwieldy, we suggest making the storyline layout more compact by only showing the relevant piece of every set curve. 
For example, each curve can start just before its respective first element in the given element order and can end just after its respective last element; see \cref{fig:variants:verticalcompactness}.

\subsection{Design Requirements}

Inspired by the set data abstract task taxonomy of Alsallakh et al.~\cite{alsallakh2016}, the discussion in Tominski et
al.~\cite{TominskiBBFMMP23}, and our own design space exploration we propose the following design requirements for StorySets:

\begin{enumerate}[leftmargin=*, label=(R\arabic*)]
    \item Connect an element to every set in which it is
      (uncertainly) contained.
    \item Discern between levels of membership certainty.
    \item Avoid overplotting and clutter.
\end{enumerate}

R1 requires that set curves must pass through all elements that are (uncertainly) contained. This explicitly shows set membership in the visual representation, thus, allowing for typical abstract set- and element-based tasks. R2 focuses on the uncertainty aspect. Different levels of uncertainty should be expressed and discernible. Furthermore, this aspect should take less visual focus than the membership representation. R3 captures a common goal for visualization optimization.

\section{Combinatorial Problem and Algorithm} 

Our proposed approach consists of two independent steps.
In Step~I, we order the elements of the ground set.
This can be captured as an instance of the traveling salesperson (TSP) problem,
which is a classic \NP-hard~\cite{HeldK70} optimization problem. 
We present three variants that use different strategies to assign edge
weights in the TSP instance.  The aim of these strategies is to reduce
the number of crossings between curves.
In Step~II, we order the set curves such that curve crossings are minimized. 
Here, we assume that the order of set elements is already fixed.
Note that the two steps can be executed independently and
alternately. Overall, our aim is to reduce the visual clutter, which is in line with our design requirement (R3).

Apart from the number of crossings, the amount of \emph{wiggle} has
also been used as a metric for storyline-type visualizations.  Wiggle
measures the vertical distance that a set curve travels; the total wiggle of
a StorySet visualization is the sum over the wiggle of the respective
set curves.  In the combinatorial problem that we are solving here,
only the vertical ordering of the set curves is important. 
Hence, the vertical distance we can measure here for a curve is the difference in positions in the vertical order of two consecutive set elements.
Similarly, in the basic variant of StorySets (see
\cref{fig:variants:glyphs3,fig:variants:roundedlines}), the distance
between consecutive set curves is fixed.  In such situations, the
amount of ``combinatorial'' wiggle is proportional to the number of
crossings.  This can be seen as follows.  In each crossing, two set
curves change their relative position; one goes up one unit, and one
goes down one unit.  Now it is clear that twice the number of
crossings exactly equals the total amount of combinatorial wiggle (in
terms of the standard distance between consecutive set curves).
Therefore, it suffices to take into account and to minimze the number
of crossings only. 

\subsection{Definitions and Notation}

A set system $S = (U, E)$ consists of a ground set $U=\{u_1,
\dots, u_n\}$ of elements and a family $E = \{s_1,
\dots, s_m\}$ of subsets of these elements, that is, for $i\in\{1,\dots,m\}$, it holds that $s_i \subseteq U$.
In an \emph{uncertain set system}~$S$, we additionally have a set of
uncertainty levels $L \subseteq [0,1]$ and a mapping $\beta \colon E
\times U \rightarrow L$, where each element $u \in U$ is contained in
each set $s \in E$ with (un-)certainty $\beta(s,u)$.
If $\beta(s,u)=0$, we know that $u \not\in s$, and if
$\beta(s,u)=1$, we know that $u \in s$. 
For values strictly between $0$ and $1$, there is uncertainty as to
whether $u \in s$ or not.  In other words, we can interpret
$\beta(s,u)$ as the probability that $u \in s$.
Such uncertain set systems are also known as \emph{fuzzy} set systems.
We are primarily interested in a discrete set of uncertainty levels $L
= \{l_1, \dots, l_k\}$, where $0=l_1 < l_2 < \dots < l_{k-1} < l_k=1$
for small values of $k$.  For example, in natural language, the set
$L=\{0, 0.25, 0.5, 0.75, 1\}$ roughly corresponds to the levels ``not
contained'', ``rather not contained'', ``maybe contained'', ``rather
contained'', ``contained''.

In addition to the uncertain set membership introduced above, we can
also define uncertain containment between two sets, based on the
uncertain set memberships given by $\beta$.
Let $s$ and $s'$ be two sets in $E$.  Then we say that $s$ is an
\emph{uncertain subset} of $s'$ if $\beta(s,u) \le \beta(s',u)$ for
every element $u \in U$.  Note, however, that uncertain containment
does not guarantee actual containment.

We encode an uncertain set system as an $m \times n$ matrix $B = (b_{i,j})$
where the sets are represented in the rows and the elements in the columns.
For $1 \le i \le m$ and $1 \le j \le n$,
$b_{i,j}=p$ means that $\beta(s_i,u_j)=l_p$.  In other words,
$s_i$ is in bin~$p$ of element~$j$.  Given this encoding of the set
system, we represent the elements by their index set
$[n]:=\{1,2,\dots,n\}$, the family of subsets by~$[m]$, and
the set of uncertainty levels by~$[k]$.

\subsection{Step~I: Ordering the Elements (or Dimensions)}
\label{sec:ordering_dimensions}

We model the problem of finding an order of the elements as a TSP problem.
This approach is commonly applied in similar contexts such as parallel coordinate plots~\cite{HeinrichW13} or linear diagrams~\cite{dobler2023crossing}. 
By encoding the similarity of the elements as edge weights between the
corresponding vertices of a complete graph, a (near-) optimal tour
corresponds to a (near-) optimal order of the elements. 

Let $B^\prime=(B,0^m)$ be the matrix that we get if we append a
column of zeros to the right of~$B$.  This all-zero column represents
a dummy vertex that determines the start and end of the order.
(Modelling the problem without the dummy vertex corresponds
to finding a solution for the star plot variant.)
Next, we construct the complete graph $G(B^\prime)$ whose vertices
$v_1,v_2,\dots,v_{n+1}$ correspond to the elements of~$U$ (including
the dummy vertex) or, equivalently, to the columns
of~$B^\prime$.  Additionally, we compute a weight matrix $W = (w_{i,j})$.
The weight $w_{i,j}$ represents the cost of traversing the
edge~$v_iv_j$ of $G(B^\prime)$ in a TSP tour.  All weights of edges
that involve the dummy vertex are set to zero.
To determine the other weights, we propose three approaches. 

\paragraph{Hamming Distance} 
First, we propose using the Hamming distance between columns~$i$
and~$j$ of~$B'$ as the weight of the edge $v_iv_j$, that is,
$w_{i,j} = \| B'_i - B'_j \|$, where $B'_i$ denotes the $i$-th column of $B'$.  The intuition behind this is that
the edge $v_iv_j$ has low weight if the assignment of set curves to
bins is similar in columns~$i$ and~$j$.

\paragraph{Upper Bound Estimation}
Second, we use an observation from the ordering algorithm of
curves to assign weights to edges by computing the upper bound of crossings if two elements are put next to each other in the order.  As explained in more detail in
\cref{sec:ordering_curves}, a crossing of two storylines is necessary
if there is an inversion of bin assignments. Such an inversion might
not have to occur immediately between two consecutive elements in
the order.  The following weights are an upper bound for the number of inversions:
\[
  w_{i,j} = \sum_{x=1}^{m}\sum_{y=x+1}^{m}
  \begin{cases}
    1 & \text{if } b_{x,i} \le b_{y,i} \text{ and } b_{x,j} > b_{y,j} \\
    0 & \text{otherwise}
  \end{cases}
\]
(Note that the symmetric case $b_{x,i} \ge b_{y,i}$ is omitted.)
Essentially, we are overcounting the occurrence of inversions and try
to place elements next to each other in the order if they avoid
inversions. Intuitively, this should reduce crossings. 
\Cref{fig:element_order} shows an example of how the upper bound estimates the crossings. 

\paragraph{Iterative Approach}
Third, we propose an iterative approach that performs multiple
computations of TSP tours, alternating with Step II (see \cref{sec:ordering_curves}).  We compute an initial element order
via the weight assignment strategy based on Hamming distance (see
above).  Next, we compute a vertical set order
$\Lambda = (\lambda_1, \dots, \lambda_n)$ by the procedure of Step II below.  Then we use this~$\Lambda$ to assign new weights to the complete graph $G(B')$.
\[
w_{i,j} = \sum_{x=1}^{m}\sum_{y=x+1}^{m} \begin{cases} 1 & \text{if } \lambda_i(x) < \lambda_i(y) \text{ and } \lambda_j(x) > \lambda_j(y) \\ 0 & \text{otherwise}\end{cases}
\]

\begin{figure}
     \centering
     \begin{subfigure}[b]{0.49\linewidth}
         \centering
         \includegraphics[width=\textwidth, page=2]{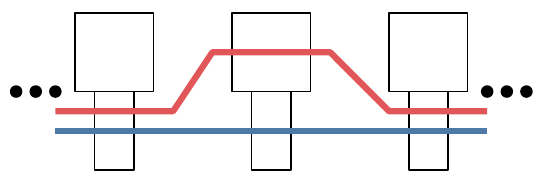}
         \caption{}
         \label{fig:crossings2}
     \end{subfigure}
     \hfill
     \begin{subfigure}[b]{0.49\linewidth}
         \centering
         \includegraphics[width=\textwidth, page=3]{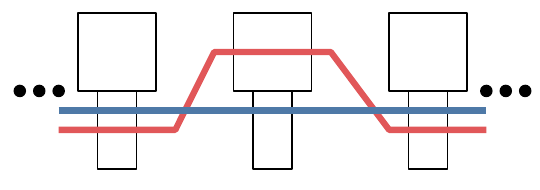}
         \caption{}
         \label{fig:crossings3}
     \end{subfigure}
     \caption{(a) and (b) show the same order of a subset of elements. Even though we do not know the actual order of set curves during Step I, we can estimate the upper bound by looking at the assigned bins of two sets. Depending on the order of remaining elements such crossings are either necessary or can be removed with the curve ordering algorithm.}
    \label{fig:element_order}
\end{figure}

(Note that the symmetric case $\lambda_i(x) > \lambda_i(y)$ is omitted.)
Now, we compute a new TSP tour with the new set of weights.  We repeat
the process of computing a TSP tour and assigning new weights until
the predefined number of iterations is reached or the number of crossings
does not go down anymore.
At first glance, this weight assignment strategy seems similar to the variant that computes the upper bound.
However, here, we compute the weights based on the previously assigned curve orders $\lambda_i$ instead of using the estimation the upper bound provides.
By knowing the actual order of set curves, we can place elements next to each other, which decreases the number of crossings.

After computing a (near-) optimal TSP tour
$(v_{i_1}, v_{i_2},\dots,v_{i_{n+1}})$, we use its vertex order
$\pi = (i_1, i_2,\dots,i_n)$ to define the order of the elements in our
StorySet.  Note that $v_{i_{n+1}}$ is the dummy vertex, which we
ignore.

\subsection{Step~II: Ordering the Curves}
\label{sec:ordering_curves}

In this section, we describe a procedure to compute the order of the
set curves in each bin while avoiding unnecessary crossings. 
The main idea of the algorithm is to process the set curves from left to right according to the element order computed in Step I.
We sort all curves in the same bin of the current element according to the order of the first two distinct bins into which they eventually separate.
In other words, if two set curves run together for a few steps, we put that one in a higher position in the higher bin once the two curves separate. 
Next, we will describe the algorithm in more detail and show that it cannot produce any unnecessary crossings.
The input is any element order~$\pi$ (which we do not change).
We assume that the columns of~$B$ are ordered according to~$\pi$.
The output is $\Lambda = (\lambda_1,\dots,\lambda_n)$, where
$\lambda_i \colon S \rightarrow [m]$ represents the order of the set
curves at each element.
For simplicity, we assume for now that, for $i\in[m]$, the curve for
set~$s_i$ is a polyline
$\langle \lambda_1(s_i), \dots, \lambda_n(s_i) \rangle$ that bends
only at the vertical element markers.
Thus, for $1 \le i < j \le m$, set curves $s_i$ and $s_j$ cross
between elements~$k$ and~$k+1$ if $\lambda_k(s_i)<\lambda_k(s_j)$ and
$\lambda_{k+1}(s_i)>\lambda_{k+1}(s_j)$, or vice versa.
Let $(b_{i,1},\dots,b_{i,n})$ be the \emph{bin vector} of~$s_i$.
Now, such an inversion is \emph{necessary} whenever a similar
condition holds for the bin assignment, namely if $b_{i,k} < b_{j,k}$
and $b_{i,k'} > b_{k,k'}$ (or vice versa) for some element $k'>k$.
Otherwise, the crossing is not necessary, and our algorithm will avoid
it.  See \cref{fig:crossings} for an example of a necessary crossing.

\begin{figure}
     \centering
     \includegraphics[width=\linewidth, page=1]{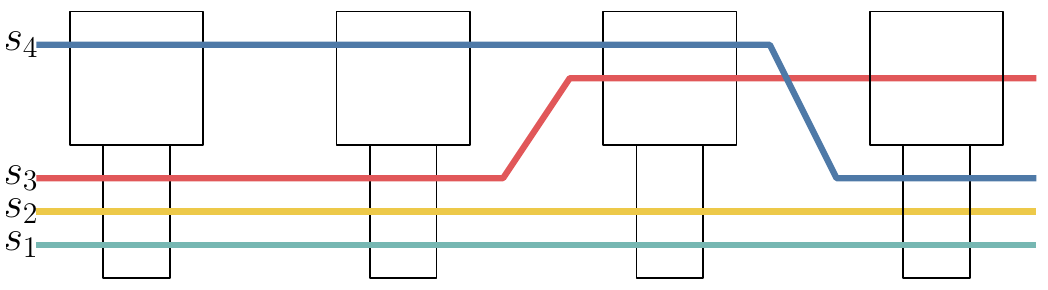}
     \caption{The case distinctions of the curve ordering algorithm. Either $s_1$ or $s_2$ needs to be kept as representative after kernelization. During lookahead, it is determined that $s_3$ must be above $s_2$ as $s_3$ switches to a different bin that is above the bin in which $s_2$ remains. The lookahead determines that initially $s_3$ must be below $s_4$. The order of $s_3$ and $s_4$ can be kept until the element in which $s_4$ switches to a bin below.}
    \label{fig:crossings}
\end{figure}

\paragraph{Algorithm}
With the above in mind, we go through the elements in the given
order~$\pi$.  For each element~$i$, we determine the order~$\lambda_i$
of the set curves.  Before that, however, we perform the following
simple kernelization.  If two or more curves have the same bin
assignment for all elements, we keep only one representative for this
subset of the curves.  This simplifies our description.  After running
our curve ordering algorithm, we insert all curves that we removed, each in
the order directly below its representative.

For the first element in~$\pi$, we perform a look-ahead to
compute~$\lambda_1$.  For each pair $s$ and $s^\prime$, we compute (in
$O(n)$ time) the first index $1 \le j \le n$ where
$b_{s,j} \not = b_{s^\prime,j}$.  If $b_{s,j} < b_{s^\prime,j}$, then
we place~$s$ below~$s^\prime$ in~$\lambda_1$, otherwise we place~$s$
above~$s^\prime$.  After we have computed all pairwise precedences, we
can simply sort the set curves, which yields~$\lambda_1$.
Computing~$\lambda_1$ in this way takes $O(nm^2)$ time.  This can be
reduced to $O(nm \log m)$ by using HeapSort and computing only the
necessary precedences on the fly.
\Cref{fig:crossings} illustrates the different cases when computing precedences between set curves. 

Then, we go through the remaining elements of~$\pi$, using a simpler
approach to determine the vertical order of the set curves.  Let $i$
be the current element.  If, for two set curves $s$ and~$s^\prime$, it
holds that $b_{s,i} \not = b_{s^\prime,i}$, we can immediately
conclude that the curve in the lower bin has to be lower in
$\lambda_i$.  Otherwise, we simply leave $s$ and $s'$ in the same
order as in $\lambda_{i-1}$.  Now sorting the curves according to
these precedences yields~$\lambda_i$ in $O(m \log m)$ time.

After processing all elements in $O(nm \log m)$ total time, we have
determined the vertical order $\Lambda=(\lambda_1,\dots,\lambda_n)$ of our StorySet.

We summarize the theoretical guarantees of our algorithm in the following theorem.

\begin{theorem}
  For a given element order, the above algorithm computes an order of
  the set curves with the minimum number of crossings.
  The algorithm runs in $O(nm \log m)$ time.
\end{theorem}

\begin{proof}
  From the description of our algorithm it is clear that it changes
  the vertical order of two set curves only if necessary.  This
  immediately yields that it produces the minimum number of crossings.
  The running time follows by adding up the running times for each
  step.
\end{proof}

\section{Quantitative Evaluation}

To evaluate the impact of different variants for computing an ordering of elements of a StorySets representation, we performed several computational experiments using samples from the Multidimensional Sexual Self-Concept Questionnaire (MSSCQ)~\cite{snell1998multidimensional}. We evaluate the resulting StorySets representations by counting the number of crossings and number of turns. We also compare the runtimes of the exact and heuristic approaches for computing a TSP tour in Step 1 of the approach. 

\subsection{Datasets}
We extract synthetic datasets from the publicly  available~\footnote{\url{http://openpsychometrics.org/_rawdata/}} dataset of responses to the MSSCQ. The MSSCQ is a self-assessment questionnaire that maps 100 Likert-type questions to 20 sexual self-concepts. The MSSCQ is commonly used in research on human sexuality. 
To create the synthetic datasets for our computational experiments, we first stripped the demographic information and sampled from the remaining 100 Likert-type questions (elements) of 17685 participants (set curves). We increased the number of elements from 5 to 100 with a step size of 5 and the number of set curves from 2 to 30 with a step size of 2. We assume that the chosen values cover all realistic data sizes, thus, providing insight on the computational limitations of the approach.
For each combination of number of elements and number of set curves we randomly sampled 5 datasets leading to a total of 1500 datasets.        

\subsection{Metrics}
We computed three quality metrics for each instance. 
First, we computed the total number of crossings ($CR$) as this intuitively correlates with the visual clutter and visual complexity in a StorySets visualization. 
Second, we computed the total number of turns ($T_\Sigma$) of the combinatorial embedding. We define a turn if a set curve $s$ changes position in the local order of two consecutive elements such that $\lambda_i(s) \not = \lambda_{i+1}(s)$.
Similar to crossings, turns should correlate with the visual complexity of a StorySets visualization.
Furthermore, both metrics should also correlate with the ink ratio, as curves that turn and cross less essentially require less ink to be drawn. 
Third, we collected runtime data, measuring total time: %
the time required to compute the order of elements and the time require to compute the local order of set curves.

\subsection{Implementation and Experimental Setup}
We implemented a prototype of StorySets in python 3.9 which we also used to run the experiments. 
As mentioned earlier, our implementation uses exact and heuristic approaches to compute the set element order.  
The heuristic approach applies the simulated annealing algorithm of NetworkX~\footnote{\url{https://networkx.org/}} to compute a short TSP tour, starting with an approximation based on the Christofides algorithm~\cite{Christofides22}. 
The exact approach uses the Concorde TSP Solver~\footnote{\url{https://www.math.uwaterloo.ca/tsp/concorde.html}} with the QSopt linear programming solver~\footnote{\url{https://www.math.uwaterloo.ca/~bico/qsopt/}} module.
We implemented the ordering algorithm for set curves in python.
For all TSP variants, besides the iterative approach, we computed the order of set curves afterwards.
Additionally, for each instance we used a random order of elements with optimal set curve ordering as a baseline case.
We ran the experiments on an Ubuntu 22.04 machine with AMD Ryzen 5600x 6-core CPU and 32 Gb RAM.  
All computation was performed on a single core of the CPU.

\subsection{Results and Discussion}

\paragraph{Determining the Number of Iterations}
In the first experiment, we evaluated the impact of the number of iterations using the iterative approach.
We set the number of iterations to 10 and computed both quality metrics after each iteration. 
Note that this is only relevant for the exact approach, as the heuristics can have worse values in all quality metrics between iterations.
For each instance we computed the iteration after which no change in any quality metric is measured.
This resulted in a mean value of $2.95$ iterations with standard deviation $\sigma=2.44$. 
We then evaluated the relative change in the quality metrics after
each iteration; see \cref{fig:iterative_exp}. %
Interestingly, the relative change regarding the quality metrics is highest in the first three iterations. 
Therefore, for input datasets of similar size to the ones used in the experiment, it is not necessary to compute more than five iterations as there is nearly no gain in quality.
We used this observation and set the number of iterations to five for the remaining experiments. 

\begin{figure}[t]
     \centering
    \begin{subfigure}[b]{0.49\linewidth}
         \centering
         \includegraphics[width=\textwidth, page=1]{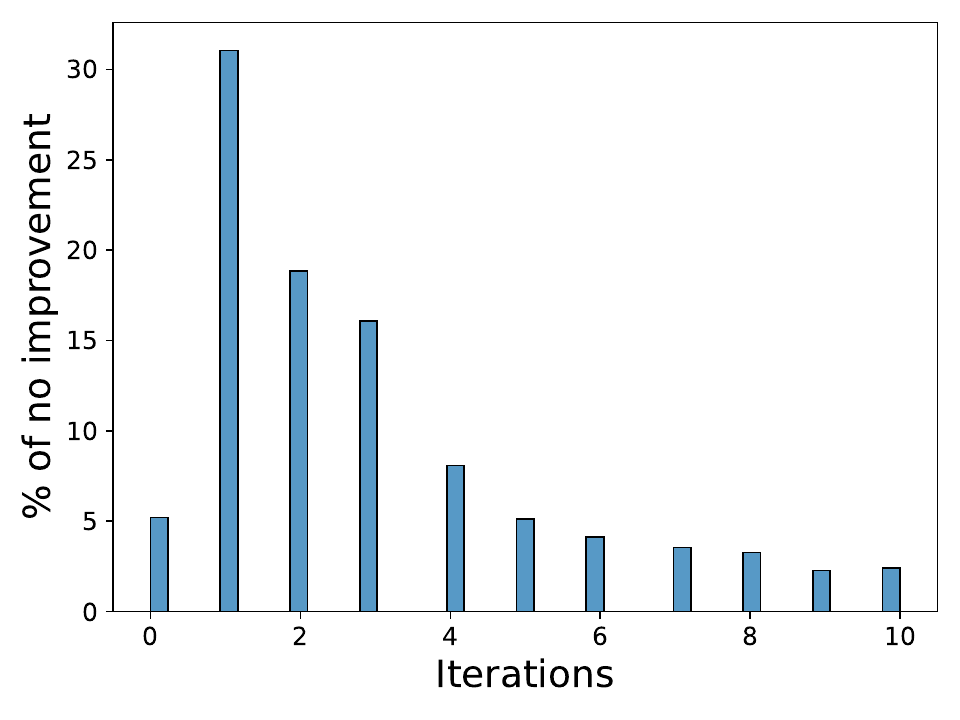}
         \caption{}
         \label{fig:iterative_exp1}
     \end{subfigure}
     \hfill
     \begin{subfigure}[b]{0.49\linewidth}
         \centering
         \includegraphics[width=\textwidth, page=1]{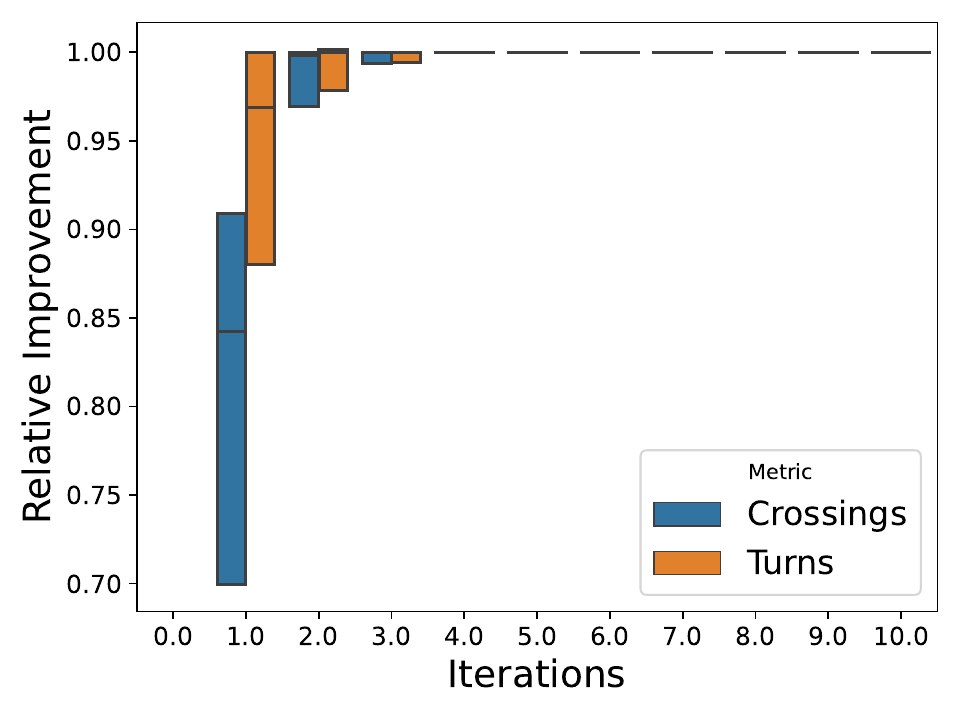}
         \caption{}
         \label{fig:iterative_exp2}
     \end{subfigure}
     \caption{Our experiments revealed that a small number of iterations is sufficient. According to~(a), most instances needed at most ten iterations before no improvement could be measured. Note that the last bin represents all instances that needed ten or more iterations. According to~(b), the first two to three iterations improved the metrics by far the most.}
    \label{fig:iterative_exp}
\end{figure}

\paragraph{Exact Approaches}
In a second experiment, we computed the quality metrics for all variants. 
For each dataset, we computed the relative number of crossings of each variant when comparing them against the baseline.
Additionally, we computed the mean of the quality metrics for all instances with equivalent sizes and variants. 
First, we compared all variants where optimal solutions for the TSP were computed.
\Cref{fig:exp_CR_opt} shows the results of the experiment for different values of $m$.
The Hamming distance as TSP weights performed worst regarding $CR$ for all combinations of instance sizes. 
However, with increasing $m$ the gap between the Hamming distance variant and all other variants shrinks.
Using the iterative approach initialized with the Hamming distance yielded a significant improvement, especially on instances with $m \le 10$.  
However, the iterative approach initialized with the Hamming distance was occasionally outperformed by the upper bound variant. especially if $m > 10$.
The iterative variant initialized with the upper bound yielded only marginal improvement.    
Thus, for $m \le 10$ and an availability of the Concorde TSP solver we suggest using the iterative approach initialized with the upper bound. For $m > 10$ the upper bound variant without additional iterations works sufficiently well.

The observed behavior regarding $T_\Sigma$ follows a similar pattern as for $CR$. All plots concerning $T_\Sigma$ can be found in the supplemental material. %

To summarize, for datasets with $m < 10$, both iterative approaches performed best regarding the quality metrics. However, for $m > 10$, using an iterative approach yields only marginal improvement over using the upper bound variant. 

\begin{figure}[t]
     \centering
         \includegraphics[width=\linewidth, page=1]{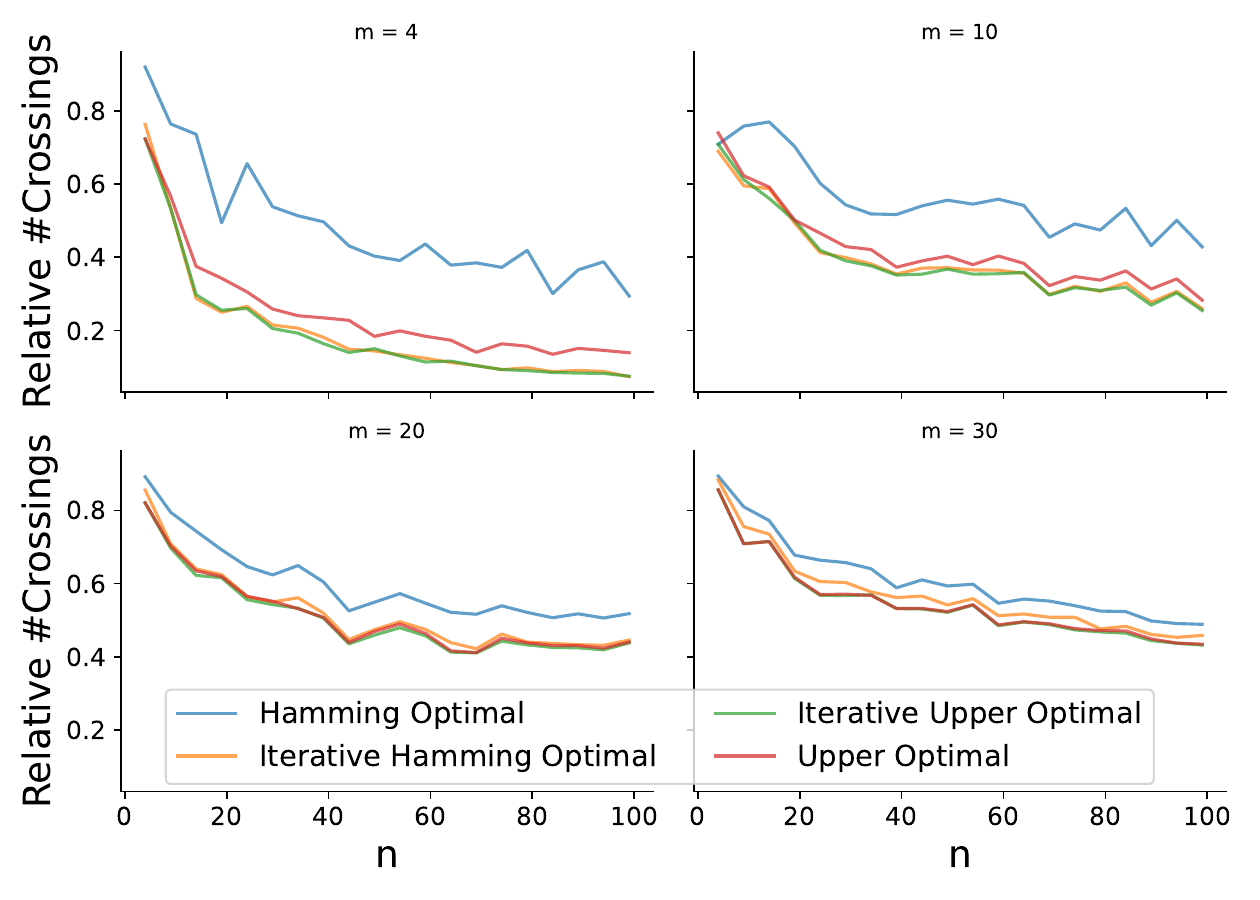}
     \caption{Comparision of number of crossings between optimal variants for different number of sets. For $m < 10$ the Hamming distance performs worse than other variants. The iterative variants occasionally outperform the upper bound variant.}
    \label{fig:exp_CR_opt}
\end{figure}

\begin{figure}[t]
     \centering
         \includegraphics[width=\linewidth, page=1]{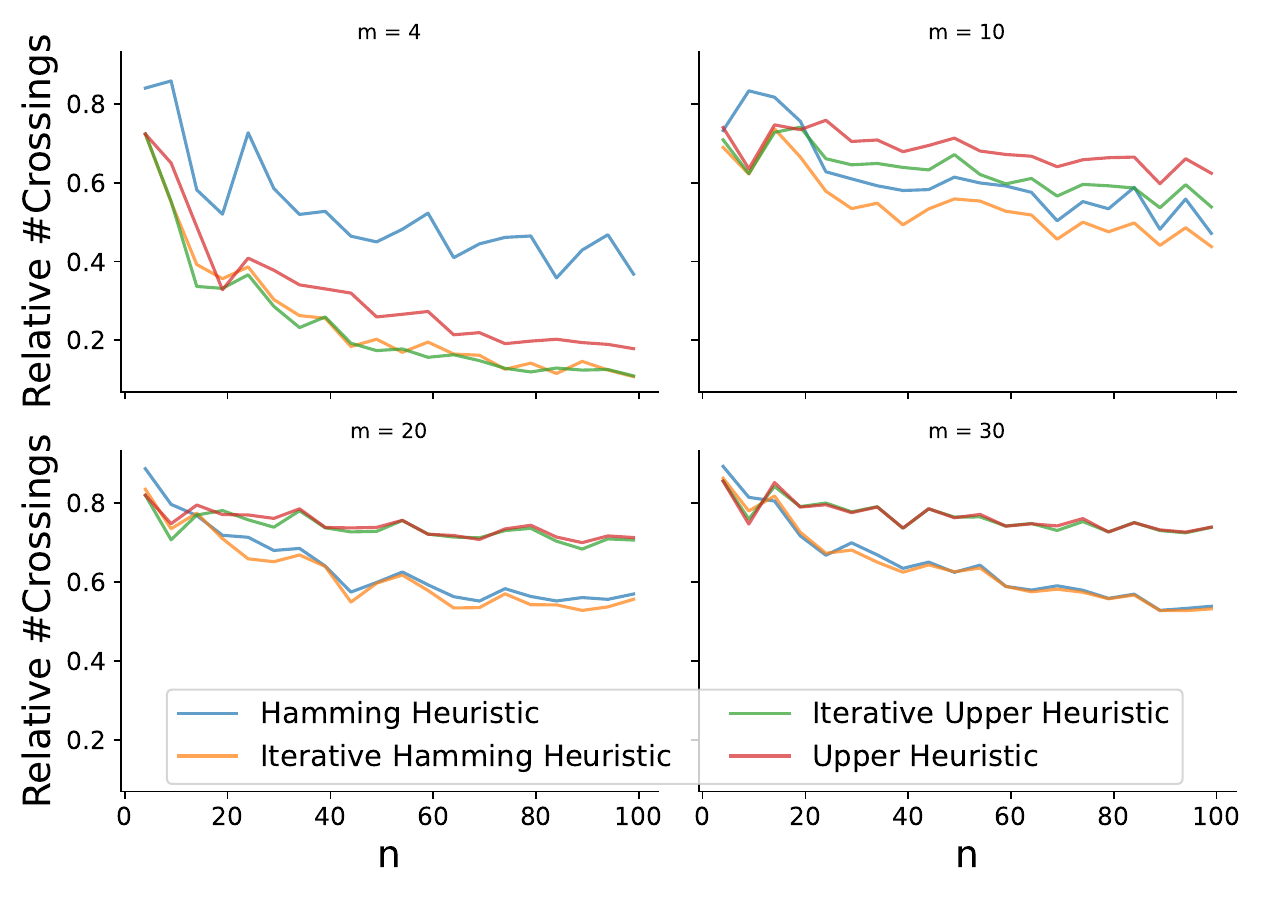}
     \caption{Comparision of number of crossings between heuristic variants for different number of sets. For $m < 10$ the Hamming distance performs worse than other variants. With increasing number of sets the other variants start to perform worse. Overall, the iterative Hamming distance variant performs consistently best.}
    \label{fig:exp_CR_heu}
\end{figure}

\paragraph{Heuristic Approaches}
We expected a similar result for the variants using heuristics to compute a TSP tour. 
\Cref{fig:exp_CR_heu} shows the experiment for different values of $m$.
However, the performance behavior is drastically different. Regarding $CR$, the iterative variant initialized with the Hamming distance performed best overall. For $m \le 6$ the Hamming distance variant performed worst, however, for $m \ge 12$ it performed second best. The upper bound variant performed slightly worse than the iterative Hamming approach for $m \le 10$; however, for $m > 10$, the performance degrades. Similarly, the iterative approach using the upper bound as initialization degrades with increasing $m$.
This leads us to believe that the heuristics struggle to find better tours with certain weight assignments and get stuck in local minima much more easily.
The observed behavior regarding $T_\Sigma$ is a similar pattern as $CR$. 

When comparing the difference in quality metrics between exact and heuristic approaches, we can report that for $n < 10 $ the heuristic approaches are able to compute a near-optimal tour. With an increasing number of elements, the gap between exact and heuristic solutions widens as seen in \cref{fig:iterative_exp_gap}. 
Furthermore, this gap is much more drastic between the different variants. 
For the Hamming distance variant, the heuristic never performs worse than $1.15$ of the optimal solution, while the iterative Hamming distance variant ($1.35$), the upper bound variant ($1.75$), and iterative upper bound variant ($1.75$) performed worse.

To summarize, the iterative approach initialized with the Hamming distance performed nearly always best for all combinations of $m$ and $n$ regarding the quality metrics. However, with an increasing number of sets, the gap closes, and the Hamming distance variant without additional iterations performed similarly. Compared to the best performance of the exact variants over all instances, both Hamming distance variants never perform worse than $1.3$ on $CR$ and $1.15$ on $T_\Sigma$. 

\begin{figure}[t]
     \centering
    \begin{subfigure}[b]{0.49\linewidth}
         \centering
         \includegraphics[width=\textwidth, page=1]{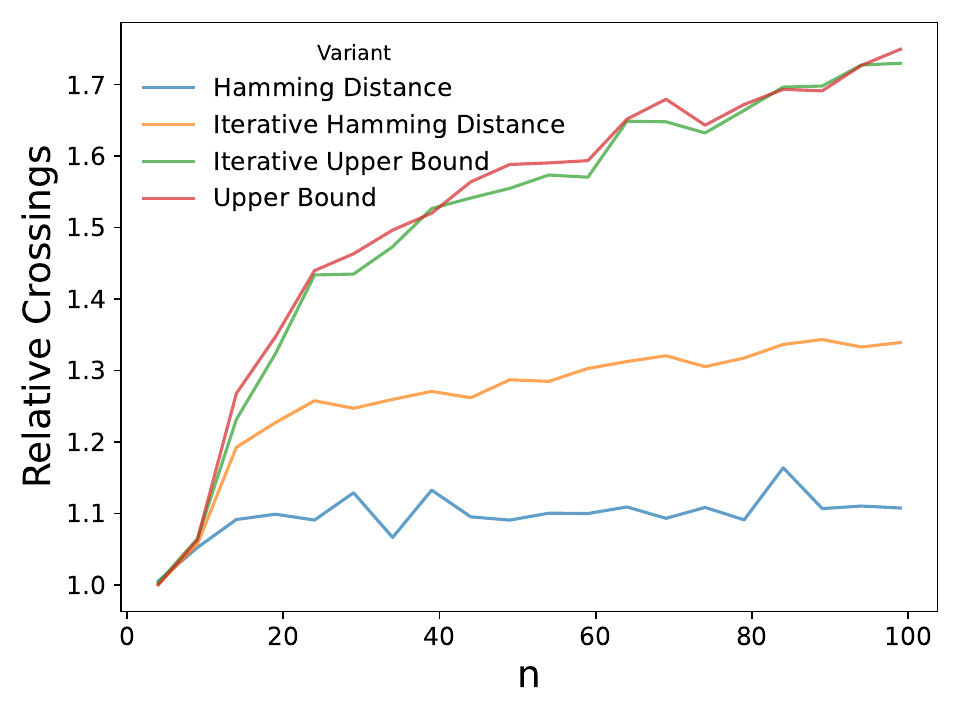}
         \caption{}
         \label{fig:iterative_exp_gap1}
     \end{subfigure}
     \hfill
     \begin{subfigure}[b]{0.49\linewidth}
         \centering
         \includegraphics[width=\textwidth, page=1]{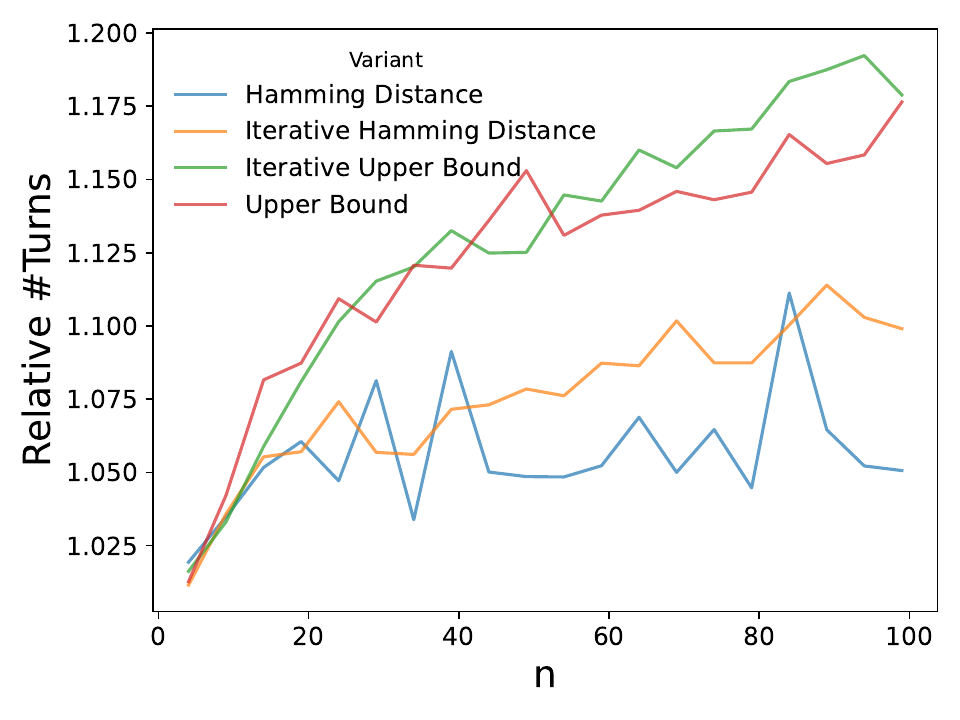}
         \caption{}
         \label{fig:iterative_exp_gap2}
     \end{subfigure}
     \caption{The relative difference of crossings (a) and turns (b) between exact and heuristic solution of the same variant. While for $n < 10$ the heuristics find a near-optimal tour, this gap widens with an increasing number of elements. Furthermore, the variance in solution quality of the heuristic approach is obvious.}
    \label{fig:iterative_exp_gap}
\end{figure}

\paragraph{Runtime}
Furthermore, our curve ordering algorithm is comparably fast in finding a TSP tour, and we confirmed that the total runtime is lopsided toward computing a TSP tour.
A comparison of the total runtimes of the different variants showed that computing the order using the Hamming distance was relatively fast, with a runtime of less than 100ms.
The variants using the upper or lower bound are slower compared to the Hamming distance variants.
Still, a StorySet visualization can be computed in less than 500ms.
The iterative approaches performed even worse but still managed to compute a StorySets visualization in less than 600ms.
Furthermore, the total runtime of the iterative variants plateaued for $n > 50$.
As $CR$ an $T_\Sigma$ does not decrease substantially as well we suspect that the iterative approaches terminate fast without further improvement.
A comparison of the runtimes between the heuristic and the exact approach shows that the efficient implementation of the exact algorithm has equal performance to the heuristic. The reason for this is mainly the efficient implementation of the Concorde TSP solver. 

In conclusion, it is feasible to compute exact solutions between $100ms$ and $500ms$ for instances that can reasonably be visualized. However, implementing the exact approach requires Concorde's efficient implementation, which might not be available in every context. For both exact and heuristic approaches, the Hamming distance had the fastest total runtime ($< 100ms$).

\section{Qualitative Evaluation}

We present two case studies to illustrate StorySets in the context of real-world data.  
The first case study shows how StorySets is applied to an uncertain set system with elements representing characters from ``The Simpsons'' with uncertain character traits as the sets themselves. The second case study focuses on visualizing the results of a recent study that evaluated the impact of Covid-19 lockdowns on mental health. Here, set elements represent character traits, and participants are set curves. 

\begin{figure*}[t]
  \centering
  \begin{subfigure}{0.67\linewidth}
      \includegraphics[width=\textwidth]{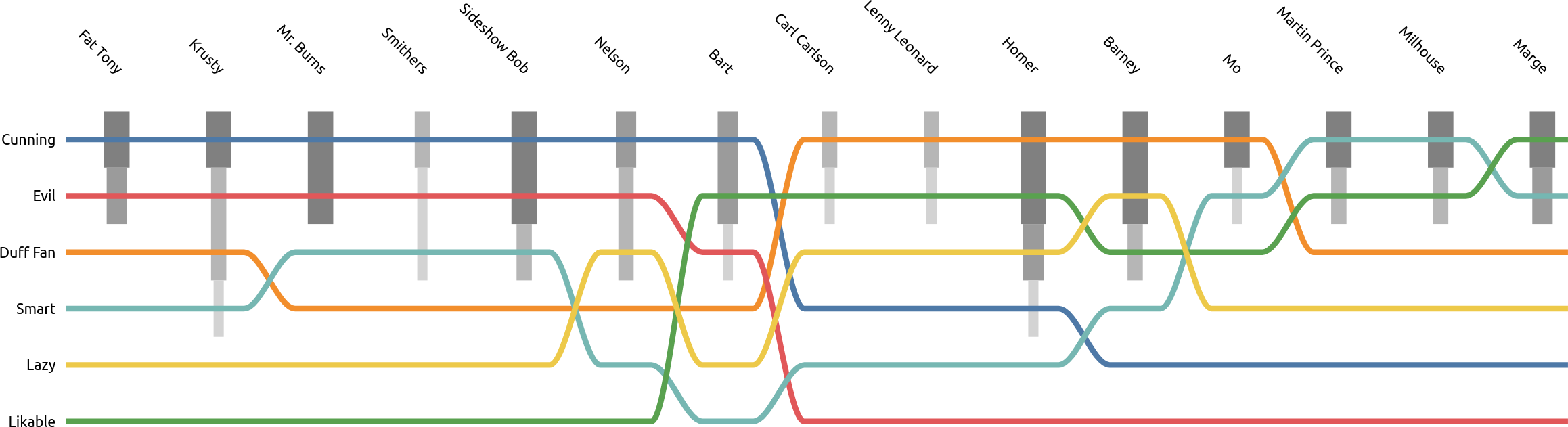}      
      \caption{}
      \label{fig:simpsons_1}
  \end{subfigure}
  \hfill
  \begin{subfigure}{0.32\linewidth}
      \includegraphics[width=\textwidth]{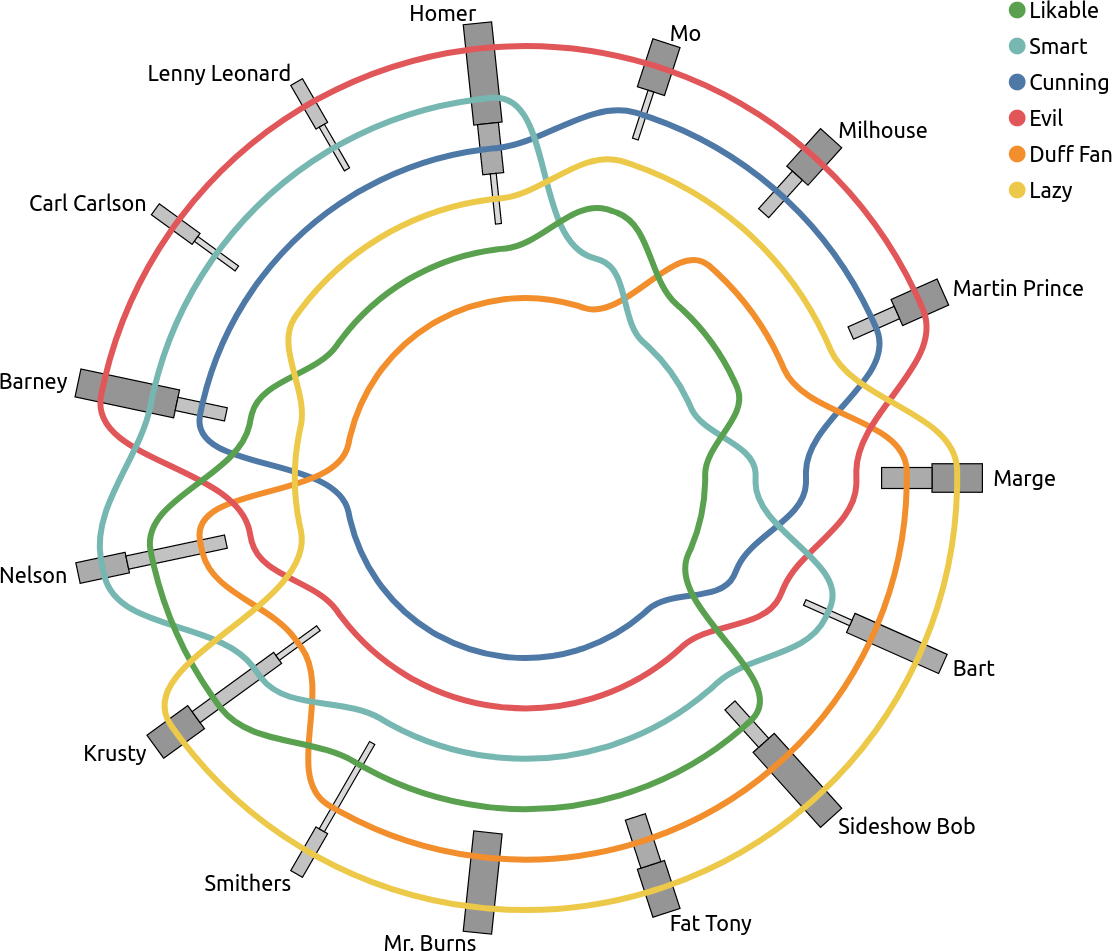}      
      \caption{}
      \label{fig:simpsons_2}
  \end{subfigure}   
  \caption{Traits of fictional characters of the tv show ``The Simpsons''. In (a) the dataset is represented with the storyline variant of StorySets. In (b) the representation uses the star variant of StorySets. In both variants it is easy to follow curves and determine set membership while membership certainty can be determined by the color and thickness of the glyphs. Furthermore, if a curve is always below (or inside) a different curve, this represents easily interpretable set containment.}
  \label{fig:usecase_simpsons}
\end{figure*}

\subsection{Simpsons Character Traits}\label{sec:use_case1}
The dataset used in this case study represents 15 main characters from the TV show ``The Simpsons'' as set elements. 
The six sets themselves represent specific character traits. 
Our dataset is based on the Open-Source Psychometrics Project~\footnote{\url{https://openpsychometrics.org/tests/characters/}} ``Which Character are You?'' quiz. 
To create the quiz, the authors asked volunteers to rate fictional characters on a scale of 0-100, where both ends of the scale were anchored with two opposing adjectives. 
Naturally, volunteers judging characters often had contradicting opinions on characters having specific traits.
Based on this observation, we handcrafted a small dataset and visualized it using the StorySets approach. 
We created two versions that highlight different variants of our approach.
\Cref{fig:simpsons_1} shows a storyline variant representation of StorySets where bin height was determined by the given local distribution in the data. 
In \cref{fig:simpsons_2} a star variant is presented.
Both variants use the iterative upper bound approach to compute the element order.
We assign a random color from the Tableau10 color scheme to set curves. 
We double encode membership certainty with the bin width and the bin color. 
Bin width is determined by scaling the width between a minimum and maximum value according to the certainty of the bin.
The color of a bin is assigned by a perceptually uniform sequential color scale, with hue corresponding to certainty.

As mentioned in the \cref{sec:uncertainty_related}, specific abstract tasks for uncertain set systems have not been defined yet. 
Thus, for our case study, we use the abstract task taxonomy of Alsallakh et al.~\cite{alsallakh2016}. 
In both examples, it is fairly straightforward to determine the set membership of an element by first locating the element and checking if and in what bin a set curve passes through the element (Design Requirement R1). 
The presence of bins also makes it easy to distinguish membership certainty between two curves passing through an element (Design Requirement R2). 
Determining which set a curve represents can easily be achieved by a lookup at the start of the curves or the legend.
For example, the character ``Sideshow Bob'' is a member of three sets, namely ``Cunning'', ``Evil'' and ``Smart''. 
While ``Cunning'' and ``Evil'' have similar certainty, ``Smart'' is less certain.      

Similarly, it is easy to count the number of sets in the system and to follow one set curve to determine which elements (or how many elements) belong to it.
For example, if we follow the orange curve representing ``Duff Fan'' we can determine that ``Lenny'', ``Carl'', ``Homer'', ``Mo'', and ``Krusty'' are members. 
Also, by comparing the color and width of the bins when following a curve, it is fairly easy to consider only elements of specific certainty.

It is also possible to analyze set-to-set relationships with a StorySets visualization. 
Set containment can be analyzed by checking if one curve is strictly below the other (or inside for the star variant). 
For example, ``Evil'' is contained in ``Cunning'' as the red line is always below the blue line. 
In the star variant, the red line is completely inside the blue line. 
We think that one of the advantages of the star variant is that checking containment relationships is similar to Euler diagrams.

To test whether two sets intersect, we can follow one of the two set curves and check whether there is an element intersected by both curves.
While it is straightforward to check simple set intersections (and their respective certainty in specific elements), it is hard to perform other more complicated analysis tasks.
In general, pair-wise intersections might be visually far apart, thus increasing the difficulty of this task when performed on a StorySets visualization. 
For example, determining the largest pairwise intersection entails following two lines while performing a pairwise comparison.

\subsection{Covid-19 and Character Strengths} 
During the Covid-19 pandemic, people were required to stay home and self-isolate, which negatively impacted mental health for many. 
Recently, a study~\cite{casali2021andra} was published that evaluated the character traits of participants and their protective role regarding mental health sustenance. The study collected responses of 944 Italian participants. 
Besides demographic information, the responses contained 120 Likert scale questions (1: ``not at all like me'' to 5: ``very much like me'') that mapped to 24 character traits. 
For example, the question ``I always complete what I begin'' mapped to perseverance.
In total, participants could have between 5 and 25 points for each character trait.
The higher the sum of points that mapped to a certain character trait, the higher a participant associates with a respective trait.
We model the 24 character traits as elements in the StorySets visualization.
A curve shows each participant's responses. 
To create bins, we first scale each character trait to the range $[0, 1]$ and assume our five bins have upper and lower bounds of $[0.0, 0.2, 0.4, 0.6, 0.8. 1.0]$.

In a typical use-case scenario, an analyst would filter out data points that are not interesting. 
Therefore, we selected all participants that had a high response value for depression.
As seen in \cref{fig:usecase_covid} we visualized this subset of participants as two StorySets visualizations in the storyline variant and a parallel coordinate plot visualization.

For both StorySets variants, we compute the exact order of elements using the iterative Hamming distance approach.
The encoding of uncertainty is similar to the previously presented case study.
For the parallel coordinate plot, we assumed that each axis has a range of $[0,1]$ and used the same order of elements as in the two StorySets examples.
For all three examples, we randomly assigned colors from the Tableau10 color palette to each set. Furthermore, all three plots have the same height.

\begin{figure*}[t]
  \centering
  \begin{subfigure}{0.85 \linewidth}
      \includegraphics[width=\textwidth]{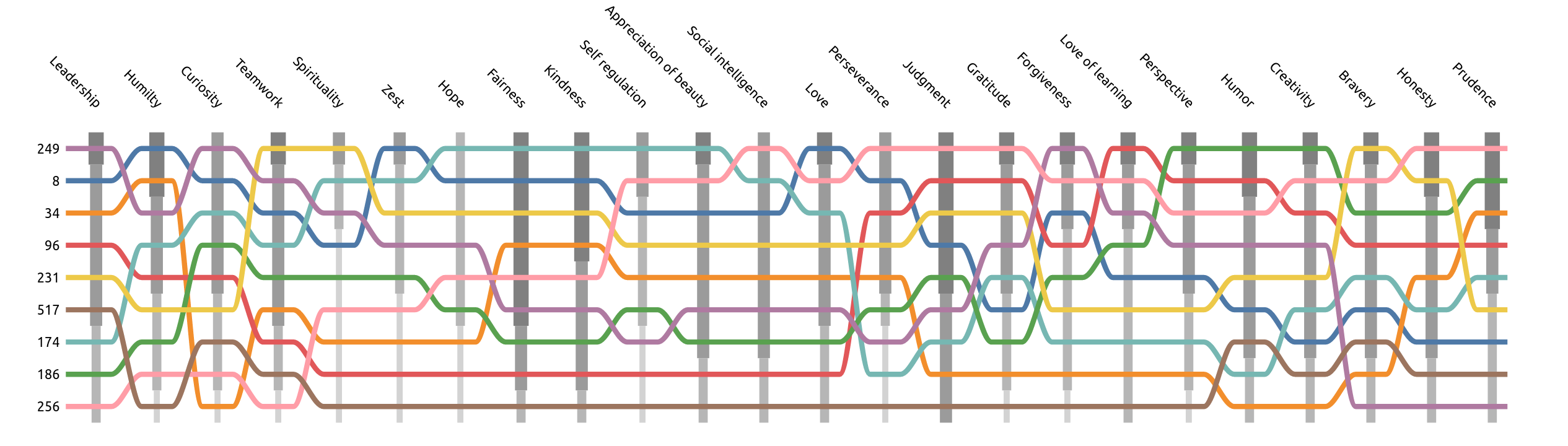}      
      \caption{}
      \label{fig:covid_1}
  \end{subfigure}
  \begin{subfigure}{0.85\linewidth}
      \includegraphics[width=\textwidth]{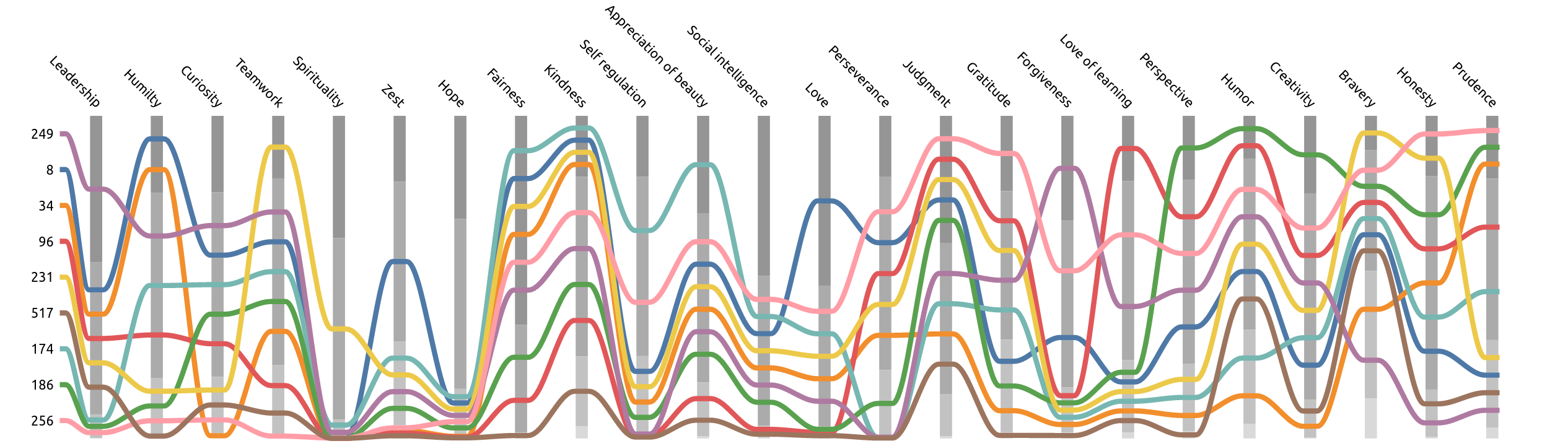}      
      \caption{}
      \label{fig:covid_2}
  \end{subfigure}
  
  \begin{subfigure}{0.85\linewidth}
      \includegraphics[width=\textwidth]{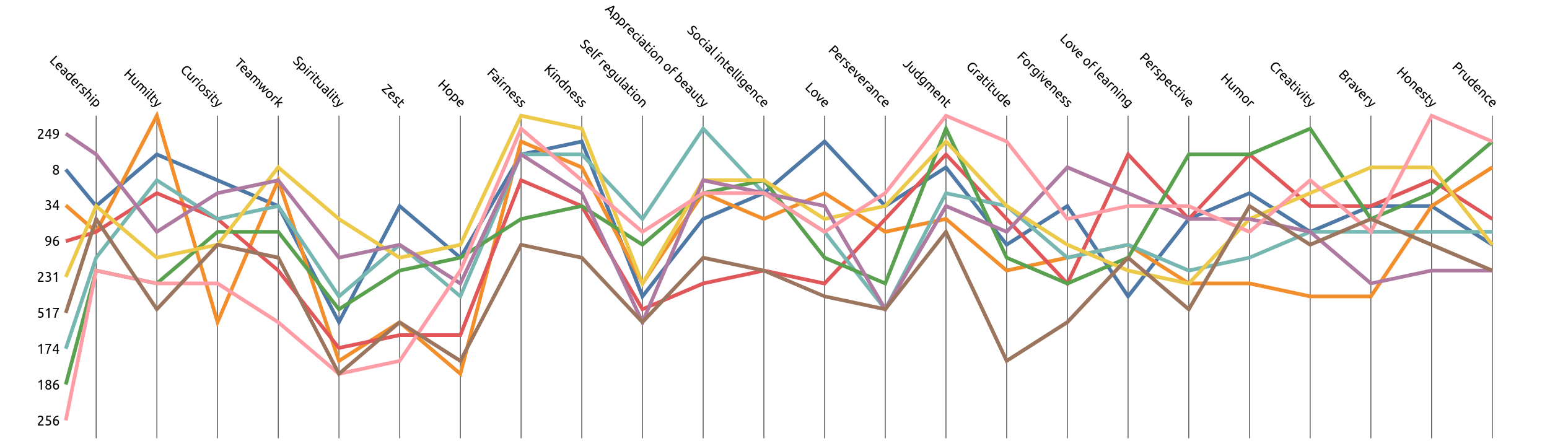}      
      \caption{}
      \label{fig:covid_3}
  \end{subfigure}
    
  \caption{Recently, a study~\cite{casali2021andra} evaluated the protective role of character traits (elements) of 944 participants (sets) during the first Covid-19 lockdown. Here, we show 10 participants with high levels of depression. In (a) a Storyline representation where bin height is scaled by the number of curves passing through. This makes following lines comparable easy and avoids overplotting similar to the parallel coordinate plot in (c). In (b) the bin height is scaled by the global distribution which uncovers the pattern that all shown participants have, compared to the general population, low values in ``Hope'' and ``Spirituality''.}
  \label{fig:usecase_covid}
\end{figure*}

An immediate observation is that the parallel coordinate plot (\cref{fig:covid_3}) creates overplotting of set curves. While it is still possible to determine the approximate coordinate of a curve on an axis, it is hard to read the actual value. This effect can even be worse if two curves have the same coordinate on subsequent axes, making it impossible to trace individual curves.
The StorySets (\cref{fig:covid_1}) approach avoids this problem and it is possible to read each individual value and trace all curves (Design Requirement R3).
The value of each curve can be determined by interpreting the color and width of a bin. 
The above-mentioned observation becomes especially noticeable between elements ``Fairness'' and ``Self-regulation''.
Here, it is difficult to distinguish between the curves in the parallel coordinate plot. 
The reasons are that the coordinates on both axes are very similar, and simultaneous contrast makes distinguishing between different colors hard.

A finding from the study was that a low value of ``Hope'', ``Spirituality'' and ``Zest'' negatively correlates with depression. While this pattern is visible in the parallel coordinate plot, it seems less clear in the StorySets approach, as this observation is encoded in the certainty of bins instead of the position of curves.
However, in \cref{fig:covid_2} we present an alternative variant of the StorySets approach, where we scale the bin height by the global distribution of the entire dataset (as discussed in the design space study). 
By scaling the bins to the global distribution, we can now see the pattern -- participants with high values for depression are on the lower spectrum when compared to the full dataset.
This is not immediately obvious in the parallel coordinate plot, where one would need to plot all curves to make this observation.
Furthermore, even though we scale the bin height to fit the global distribution, the combinatorial embedding of the StorySets visualization does not change. Thus, no additional crossings are introduced.

\section{Discussion and Limitations}

As our focus in this paper was on exploring the design space of uncertain set visualizations and on providing algorithms that support the proposed new method, it is difficult to also incorporate a human-subjects study. Such a study could evaluate whether StorySets can be used for uncertain set visualization, but this remains as future work. It would also be worthwhile to consider the impact of design space choices (e.g., uniform/data-driven bin size, element glyph representation, set curve rendering, etc.) on task performance.

Similarly, it is unclear how well the metrics from the computational experiment translate to task performance, as minimizing crossings might not be the correct choice to determine the optimal element order. 

StorySets considers only one type of uncertainty in set systems: element-set membership uncertainty. Generalizing the proposed approach to handle other types of uncertainties in set visualization remains for future work.

Intuitively, StorySets does not easily scale to a large number sets, as one of the design principles is to clearly show every set without overlaps and overplotting. 
The storyline layout of StorySets can likely handle many elements, but the star layout is clearly limited in that respect. One plausible strategy for dealing with large numbers of elements and sets is to provide summarization features for an overview and interactivity for details.

Similarly, with an increasing number of sets, it becomes harder to differentiate between different curves. 
Using only color works for a small number of sets before reaching the limitations of the human visual system. 
However, using different line styles might improve the number of sets that can be visualized at the same time.
Still, it is unlikely to reasonably visualize more than 10--20 sets. 

Finally, we barely touched on integrating interactivity in the StorySets approach. For example, highlighting a set could be used to show other intersecting sets or bring the bins the curve passes through into the focus. As we have shown in the computational experiment, computing a StorySets representation is fast. Thus, it is feasible to integrate filter-on-demand or dynamic reordering functionality.

\section{Conclusions}

We presented a design space exploration for uncertainty in set visualization and described StorySets. To the best of our knowledge, StorySets is the first method for visualizing uncertain set systems. Underlying the proposed methods is an exact algorithm for optimally ordering set-curves within each element's bins, which avoids overplotting, thus, tracing a set curve is easy. Furthermore, it easy to see set containment in both the storyline and star variants. 
Our computational experiments showed that computing an optimal order of elements and set curves is feasible in near-realtime making it possible to integrate StorySets in an interactive application. 
We also showed that StorySets can  be used to effectively visualize
multi-dimensional discrete data, avoiding the overplotting inherent in parallel coordinates.

\bibliographystyle{abbrv-doi-hyperref}

\bibliography{storysets}

\end{document}